\documentclass[psamsfonts]{amsart}
\usepackage{graphicx}

\usepackage{amssymb}
 \usepackage{mathtools}
 \usepackage{pifont}
 \usepackage{multicol}
\usepackage{mathrsfs}
\usepackage{color}
\usepackage{multicol}
\usepackage{amsthm}
\usepackage{float}
\usepackage{tikz}
\usepackage{tikz-cd}
\usetikzlibrary{arrows}
\usetikzlibrary{shapes.arrows}   
\usetikzlibrary{positioning}
\usepackage{mathrsfs}
\usepackage{amssymb}
\usepackage{mathtools}
\usepackage{chessfss}

\DeclarePairedDelimiter\floor{\lfloor}{\rfloor}
\DeclarePairedDelimiter{\ceil}{\lceil}{\rceil}
\DeclarePairedDelimiter{\abs}{\lvert}{\rvert}
\definecolor{1}{rgb}{1,0.2,0.3}
\definecolor{2}{rgb}{0.1,0.3,0.5}
\definecolor{3}{rgb}{1,1,0}
\definecolor{4}{rgb}{255,255,255}
\newtheorem{theorem}{Theorem}
\newtheorem{corollary}[theorem]{Corollary}
\newtheorem{lemma}[theorem]{Lemma}

\theoremstyle{definition}

\theoremstyle{remark}

\newtheorem{question}{Open Question}

\newcommand{\psat}{\ensuremath{\mathrm{P3SAT}_{\overline{\underline{3}}}}}

\begin{document}

\tikzset
{
  x=.23in,
  y=.23in,
}

\title{Art gallery problem with rook and queen vision}
\author{Hannah Alpert and \'Erika Rold\'an}
\address{The Ohio State University, 231 W 18th Ave, Columbus, OH 43202}
 
\date{\today}
\maketitle

\begin{abstract}
How many chess rooks or queens does it take to guard all the squares of a given polyomino, the union of square tiles from a square grid?  This question is a version of the art gallery problem in which the guards can ``see'' whichever squares the rook or queen attacks.  We show that $\floor{\frac{n}{2}}$ rooks or $\floor{\frac{n}{3}}$ queens are sufficient and sometimes necessary to guard a polyomino with $n$ tiles.  We also prove that finding the minimum number of rooks or the minimum number of queens needed to guard a polyomino is NP-hard.  These results also apply to $d$--dimensional rooks and queens on $d$--dimensional polycubes.  We also use bipartite matching theorems to describe sets of non-attacking rooks on polyominoes.
\end{abstract}
\section{Introduction}\label{sec:introduction}

Perhaps the most famous mathematical problem involving chess pieces is the eight-queens problem from the mid-1800s, which asked for all the ways to place eight mutually non-attacking queens on a standard chessboard; generalizations of this problem have been popular research problems ever since that time~\cite{bell2009survey}.  In particular, a recent NP-hardness result addresses the problem of when an initial placement of queens can be extended~\cite{gent2017complexity}.  (In Theorem~\ref{teo:extend} we show that the analogous problem for rooks instead of queens is polynomial-time solvable, even on non-rectangular boards.)

In our paper instead of placing chess pieces on rectangular chessboards, we use arbitrary polyominoes as our boards.  A \textbf{\textit{polyomino}} is the union of finitely many squares from the infinite chessboard, that is, from the standard tiling of the plane by unit squares, such that the interior of the polyomino is connected.  We refer to the squares as tiles.  Similarly, in $d$ dimensions, a \textbf{\textit{$d$--polycube}} is the union of finitely many $d$--cubes (tiles) from the standard tiling of $\mathbb{R}^d$ by unit $d$--cubes, such that the interior of the polycube is connected.

In this paper we place chess rooks and queens on the tiles to guard the polyomino or polycube.  (One could also consider the same questions for other chess pieces, but rooks and queens share the ability to attack along rows and columns.)  Roughly, we say that our rook or queen's line of attack ends when it crosses outside the polycube.  More precisely, we imagine the $d$--cubes centered at the points of $\mathbb{Z}^d$.  Suppose that we have a rook at $(0, \ldots, 0)$.  For each point that has all coordinates $0$ except for one coordinate $\pm 1$, we say that the $d$--dimensional rook \textbf{\textit{guards}} or \textbf{\textit{attacks}} the tiles with coordinates given by all natural-number multiples of this point such that all the smaller natural-number multiples are tiles of our $d$--polycube.  Similarly, suppose that we have a queen at $(0, \ldots, 0)$.  For each point that has all coordinates equal to $0$ or $\pm 1$, we say that the $d$--dimensional queen guards or attacks the tiles with coordinates given by all natural-number multiples of this point such that all the smaller natural-number multiples are tiles of our $d$--polycube.

This paper addresses the problem of finding the minimum number of rooks or queens needed to guard all the tiles of a given polyomino.  A set of rooks or queens that guards all the tiles of a polyomino is also called a \textbf{\textit{dominating}} set.  The question is similar to the famous art gallery problem of Chv\'atal, in which $\floor{\frac{n}{3}}$ guards are sufficient to see all points in an $n$--vertex polygon, and some polygons need this many guards~\cite{Chvatal75}.  The art gallery problem for polyominoes was introduced in~\cite{Biedl12}, where those authors show that $\floor{\frac{n+1}{3}}$ guards are sufficient and sometimes necessary to see all points in an $n$--tile polyomino.  Different notions of vision are possible in art gallery problems.  Classically the guards are points that can see along line segments in all directions; in our setup, the guards use ``rook vision'' that sees horizontally and vertically from all the points of a given tile, or ``queen vision'' that also sees diagonally.  Our Theorems~\ref{teo:minimum-rook} and~\ref{teo:minimum-queen} show that $\floor{\frac{n}{2}}$ rook guards, or $\floor{\frac{n}{3}}$ queen guards, are sufficient and sometimes necessary to guard an $n$--tile polyomino with rook vision.  For square chessboards, asymptotic results are known on the number of queens needed to dominate~\cite{Grinstead90}.

The computational complexity of finding the minimum number of guards for a given art gallery has been studied for many variations on the original art gallery setup.  The original version is NP-hard~\cite{Lee86}, and so is the problem on polyominoes with standard vision, even if the polyomino is required to be bounded by a simple closed curve (i.e., it has no holes, or is simply connected)~\cite{Biedl11}.  

Another notion of vision that is useful for computation on polyominoes is \textbf{\textit{$r$--visibility}}, in which one point can see another if and only if the rectangle aligned with the axes with those two points as opposite corners is completely contained within the polyomino.  The problem of finding the minimum number of $r$--visibility guards for a given polyomino is polynomial-time solvable if the polyomino is required not to have holes~\cite{Worman07}.  However, Iwamoto and Kume show that $r$--visibility guard set problem on polyominoes with holes is NP-hard~\cite{Iwamoto14}; we adapt their proof to prove Theorems~\ref{teo:nphard-rook} and~\ref{teo:nphard-queen}, which say that finding the minimum number of rooks or queens to guard a given polyomino is NP-hard.

Theorem~\ref{teo:nphard-rook}, about the NP-hardness of the art gallery problem for rook vision, should be contrasted with a very similar art gallery problem.  In the paper~\cite{Ntafos86} of Ntafos (described in the book~\cite{ORourke87}), a grid consists of a finite union of horizontal and vertical line segments.  With the usual notion of vision, any guards must then see horizontally and vertically as rooks do, but the problem of finding the minimum number of guards is polynomial-time solvable~\cite{Ntafos86}.  This is because the two problems are slightly different.  For example, a 2 by 3 rectangle polyomino can be guarded by two rooks, but the corresponding grid of 2 horizontal and 3 vertical segments requires three guards, because every point on every segment must be guarded.  The grid art gallery problem in dimension at least 3, though, is NP-hard~\cite{Ntafos86}.

\subsection{Our results}

Our first pair of results are analogous to the result in~\cite{Biedl12} that gives the minimum number of guards needed to see all of an $n$--tile polyomino.  The difference is that our guards are rooks or queens rather than points with standard vision.


\begin{theorem}\label{teo:minimum-rook}
In any dimension $d$, the number of $d$--dimensional rooks that are sufficient and sometimes necessary to guard a $d$--polycube with $n$ tiles is $\min\left\{1, \floor{\frac{n}{2}} \right\}$.
\end{theorem}

\begin{theorem}\label{teo:minimum-queen}
In any dimension $d$, the number of $d$--dimensional queens that are sufficient and sometimes necessary to guard a $d$--polycube with $n$ tiles is  $\min\left\{1, \floor{\frac{n}{3}} \right\}$.
\end{theorem}

Our second pair of results are analogous to the result in~\cite{Iwamoto14} that shows that finding the minimum number of $r$--visibility guards for a polyomino is NP-hard.  Instead of using $r$--visibility we adapt the proof for rook vision and queen vision.  Formally, we say that an instance of the rook-visibility guard set problem for polyominoes is a pair $(P,m)$ where $P$ is a polyomino and $m$ is a positive integer. The problem asks whether there exists a set of $m$ rooks placed in $P$ which guard all tiles of $P$.  A similar definition applies to $d$--polycubes for any dimension $d$; proving that the rook-visibility guard set problem for polyominoes is NP-hard immediately implies the result for $d$--polycubes for any $d$, because we can restrict to those polycubes that stay in a single $2$--dimensional layer.  Likewise, we can define the queen-visibility guard set problem for polyominoes and for $d$--polycubes for any $d$.

\begin{theorem}\label{teo:nphard-rook}
The rook-visibility guard set problem for polyominoes is NP-hard, and thus the rook-visibility guard set problem for $d$--polycubes is NP-hard for any dimension $d$.
\end{theorem}

\begin{theorem}\label{teo:nphard-queen}
The queen-visibility guard set problem for polyominoes is NP-hard, and thus the queen-visibility guard set problem for $d$--polycubes is NP-hard for any dimension $d$.
\end{theorem}

In Section~\ref{sec:thm1} we prove Theorems~\ref{teo:minimum-rook} and~\ref{teo:minimum-queen}. In Section~\ref{sec:thm2} we prove Theorems~\ref{teo:nphard-rook} and~\ref{teo:nphard-queen}.  In Section~\ref{sec:nonattacking} we prove several more theorems about non-attacking rooks on polyominoes by viewing non-attacking rook sets on polyominoes as matchings in bipartite graphs; this construction is specific to rooks and specific to $2$--dimensional polyominoes.  The paper concludes with some open questions.

\medskip

\emph{Acknowledgments.}  H.~Alpert is supported by the National Science Foundation under Award No.~DMS 1802914, and this project was begun when both authors were in residence at ICERM (Institute for Computational and Experimental Research in Mathematics) in Autumn 2016.  We would like to thank L\'aszl\'o Kozma, who suggested the proof of Theorem~\ref{teo:minimum-rook} that appears in the paper, and Frank Hoffmann, who brought the paper~\cite{Ntafos86} of Ntafos to our attention.


\section{Polyominoes Most Difficult to Guard}\label{sec:thm1}

In this section we prove Theorems~\ref{teo:minimum-rook} and~\ref{teo:minimum-queen}.

\begin{proof}[Proof of Theorem~\ref{teo:minimum-rook}]
First we show that every $d$--polycube with $n$ tiles can be guarded by $\floor{\frac{n}{2}}$ $d$--dimensional rooks.  If the polycube has only one tile, we place one rook on that tile.  Otherwise, we $2$--color the polycube according to the parity of the sum of coordinates of each tile; that is, tiles that share a $(d-1)$--dimensional face get opposite colors.  We take the color with the smaller number of tiles, and place rooks on all tiles of that color, so that there are at most $\floor{\frac{n}{2}}$ rooks.  Because the interior of the polycube is connected, every tile shares a $(d-1)$--dimensional face with some other tile.  Thus every tile is guarded by at least one of the rooks.

Next we exhibit, for every natural number $n$, a polyomino with $n$ tiles such that the minimum number of rooks needed to guard it is $\floor{\frac{n}{2}}$. To get an example in any dimension $d$, we can just thicken this polyomino to $d$ dimensions.  The construction, shown in Figure~\ref{fig:leader} for $n=10$ and $n=11$ tiles, consists of one center column with individual tiles attached on either side, alternating between left and right.  More precisely, given $n=2m$ tiles it is possible to form $m$ horizontal dominoes, that we denote by $d_1$, $d_2$,..., $d_m$, that we use to construct a polyomino $P_n$ by stacking them vertically in the following way: place the first domino $d_1$, then place on top of $d_1$ the second horizontal domino $d_2$ in such a way that the left tile of $d_2$ is on top of the right tile of $d_1$; then, place $d_3$ on top of $d_2$ in such a way that the right tile of $d_3$ is on top of the left tile of $d_2$; keep placing the rest of the dominoes in the same way, alternating the right and left placement relationships from one added domino to the next one. This construction generates a polyomino $P_n$ for even $n = 2m$ such that each domino has one tile in a common column that we denote by $c^*$. For an odd number $n=2m-1$ we construct a polyomino $P_n$ by placing a tile on the bottom-most part of column $c^*$ of $P_{n-1}$.  To prove that $\left\lfloor \frac{n}{2} \right\rfloor$ rook guards are needed, we observe that there are exactly $\left\lfloor \frac{n}{2} \right\rfloor$ tiles not in the center column, and that no two of these can be guarded by the same rook.
\end{proof}

\begin{figure}[h]
\begin{center}
\begin{tikzpicture} 
\foreach \x/\y in {  0/1, -1/1, 0/2, 1/2, 0/3, -1/3, 0/4, 1/ 4, 0/5, -1/5} { 
\path [draw=gray, fill=gray] (.5+\x-0.45, .5+\y-0.45)
-- ++(0,.9)
-- ++(.9,0)
-- ++(0,-.9)
--cycle;
}

\foreach \x/\y in { 0/1, 1/2, -1/3, 1/4, -1/5} { 
 \node[anchor=west] at (\x+.05, \y+.5) {\rook};
}

\foreach \x/\y in { 0 / 0, 0/1, -1/1, 0/2, 1/2, 0/3, -1/3, 0/4, 1/ 4, 0/5, -1/5} { 
\path [draw=gray, fill=gray] (10.5+\x-0.45, .5+\y-0.45)
-- ++(0,.9)
-- ++(.9,0)
-- ++(0,-.9)
--cycle;
}

\foreach \x/\y in { 0/1, 1/2, -1/3, 1/4, -1/5} { 
 \node[anchor=west] at (\x+10.05, \y+.5) {\rook};
}

\end{tikzpicture}
\end{center}
\caption{On the left, we depict a 10-omino that cannot be guarded with fewer than $5$ rooks. On the right, we depict an $11$-omino that cannot be guarded with fewer than $5$ rooks. Any polyomino with this kind of shape with $n$ tiles cannot be guarded with fewer than $\left\lfloor \frac{n}{2} \right\rfloor$ rook guards.}
\label{fig:leader}
\end{figure}
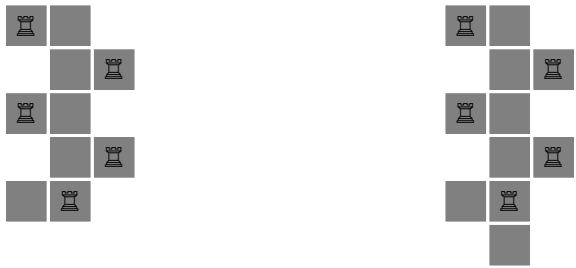

The proof of Theorem~\ref{teo:minimum-queen} relies on labeling the tiles of our $d$--polycube according to the number of steps needed to get back to a given root tile.  We say that two tiles are \textbf{\textit{adjacent}} or are \textbf{\textit{neighbors}} if they share a $(d-1)$--dimensional face.  The \textbf{\textit{$\ell^1$--distance}} between any two tiles of a $d$--polycube is the minimum number of steps needed to get from one to the other, such that each step goes from one tile of the polycube to an adjacent tile of the polycube.

\begin{lemma}\label{lem:dist-2}
In a $d$--polycube, any queen must guard every tile of $\ell^1$--distance at most $2$ from that queen.
\end{lemma}

\begin{proof}
Suppose our queen is at $(0, \ldots, 0)$.  It certainly guards every tile of $\ell^1$--distance $1$ from it, because those tiles are adjacent to it.  Consider any tile of $\ell^1$--distance $2$ from the queen.  If it has all coordinates $0$ except for two that are $\pm 1$, then we know that the queen guards that tile.  If it has all coordinates $0$ except for one that is $\pm 2$, then we just need to check that the one tile between our chosen tile and the queen is inside the polycube.  But, if it were not, then our chosen tile could not have $\ell^1$--distance $2$ from the queen, because every path that stays within the polycube would have to go around the missing tile.  These are the only possibilities for the tiles within $\ell^1$--distance $2$ of the queen, and the queen guards all of them.
\end{proof}

Because of this lemma, when guarding a $d$--polycube with queens it suffices to find a placement of the queens such that every tile of the polycube is within $\ell^1$--distance $2$ of a queen.

\begin{proof}[Proof of Theorem~\ref{teo:minimum-queen}]
First we show that every $d$--polycube with $n$ tiles can be guarded by $\floor{\frac{n}{3}}$ queens.  We select one tile of the polycube to be the \textbf{\textit{root}}, and label all the tiles according to their $\ell^1$--distance from the root.  If possible, we select as the root a tile that is adjacent to only one other tile.  If every tile of the polyomino is adjacent to more than one other tile, then we may select any tile as the root.

The tiles are partitioned into three sets according to whether their $\ell^1$--distance from the root is $0$, $1$, or $2$ mod $3$.  If the $2$ mod $3$ set is empty, then a queen placed on the root guards the whole polyomino.  Otherwise, we place queens on all the tiles in the smallest of the three sets.  We claim that every tile is within $\ell^1$--distance $2$ of at least one queen.  If the queen set is the $0$ mod $3$ set, then from an arbitrary tile, we can find a queen within two steps by walking along a shortest path to the root.  Similarly, if the queen set is the $1$ mod $3$ set, then from every tile except the root, we can find a queen within two steps by walking along a shortest path to the root, and we know that the root is also adjacent to a queen.  

If the queen set is the $2$ mod $3$ set, then from every tile of $\ell^1$--distance at least $2$ from the root, we can find a queen within two steps toward the root, so it remains to check the root and the tiles adjacent to the root.  We know that the root is $\ell^1$--distance exactly $2$ from a queen.  If it is adjacent to only one tile, then that tile is adjacent to a queen.  Otherwise, every tile in the polyomino is adjacent to at least two tiles.  Given a tile adjacent to the root, it must also be adjacent to another tile, and that tile is of $\ell^1$--distance $2$ from the root and thus has a queen---for parity reasons, no two tiles adjacent to the root are adjacent to each other.  Thus in all cases, every tile is within $\ell^1$--distance $2$ of a queen, and thus by Lemma~\ref{lem:dist-2} it is guarded by a queen.

We exhibit a polyomino with $n$ tiles that needs $\floor{\frac{n}{3}}$ queen guards, and this polyomino may be fattened to any dimension $d$.  The construction is very similar to the construction for rooks and is shown in Figure~\ref{fig:ladder-queen}.  If $n = 3m$, we make $m$ rows of $3$ tiles each, and stack them so that the center column contains the right-most tile of the first row, the left-most tile of the second, and the right-most tile of the third, and continues to alternate.  Then if $n = 3m + 1$ or $n = 3m + 2$, we add the remaining one or two tiles to the bottom of the center column.  Then no two of the $m$ tiles furthest to the left and right can be guarded by the same queen, so at least $m = \floor{\frac{n}{3}}$ queens are needed to guard this polyomino.
\end{proof}

\begin{figure}[h]
\begin{center}
\begin{tikzpicture} 
\foreach \x/\y in {  0/1, -1/1, -2/1, 0/2, 1/2, 2/2, 0/3, -1/3, -2/3, 0/4, 1/ 4, 2/4, 0/5, -1/5, -2/5} { 
\path [draw=gray, fill=gray] (.5+\x-0.45, .5+\y-0.45)
-- ++(0,.9)
-- ++(.9,0)
-- ++(0,-.9)
--cycle;
}

\foreach \x/\y in { 0/1, 2/2, -2/3, 2/4, -2/5} { 
 \node[anchor=west] at (\x+.05, \y+.5) {\queen};
}

\foreach \x/\y in {  0/1, -1/1, -2/1, 0/2, 1/2, 2/2, 0/3, -1/3, -2/3, 0/4, 1/ 4, 2/4, 0/5, -1/5, -2/5, 0/0} {
\path [draw=gray, fill=gray] (6.5+\x-0.45, .5+\y-0.45)
-- ++(0,.9)
-- ++(.9,0)
-- ++(0,-.9)
--cycle;
}

\foreach \x/\y in { 0/1, 2/2, -2/3, 2/4, -2/5} { 
 \node[anchor=west] at (\x+6.05, \y+.5) {\queen};
}

\foreach \x/\y in {  0/1, -1/1, -2/1, 0/2, 1/2, 2/2, 0/3, -1/3, -2/3, 0/4, 1/ 4, 2/4, 0/5, -1/5, -2/5, 0/0, 0/-1} {
\path [draw=gray, fill=gray] (12.5+\x-0.45, .5+\y-0.45)
-- ++(0,.9)
-- ++(.9,0)
-- ++(0,-.9)
--cycle;
}

\foreach \x/\y in { 0/1, 2/2, -2/3, 2/4, -2/5} { 
 \node[anchor=west] at (\x+12.05, \y+.5) {\queen};
}

\end{tikzpicture}
\end{center}
\caption{For every $n$, there is a polyomino with $n$ tiles that cannot be guarded by fewer than $\floor{\frac{n}{3}}$ queen guards.  Examples for $n = 15, 16, 17$ are shown.}
\label{fig:ladder-queen}
\end{figure}
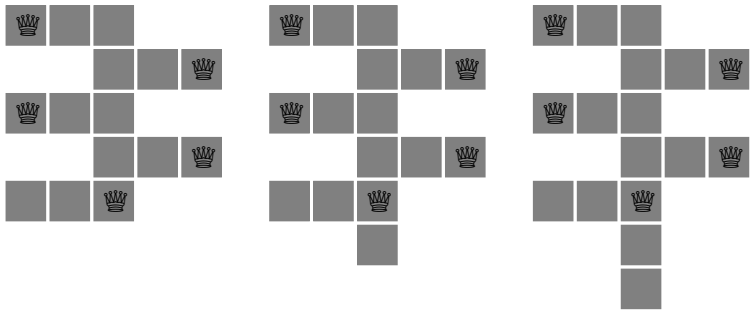


\section{NP-Hardness Proof}\label{sec:thm2}

In this section we prove Theorems~\ref{teo:nphard-rook} and~\ref{teo:nphard-queen} about the NP-hardness of the rook-visibility and queen-visibility guard set problems, using the NP-hardness of PLANAR 3SAT.  First we prove the rook version, and then we show how to make the minor modifications needed for the queen version.  

We reproduce here the definition of PLANAR 3SAT as given by Iwamoto and Kume in~\cite{Iwamoto14}, who attribute the description there to~\cite{Garey79}.  Let $U = \{x_1, x_2, \ldots, x_n\}$ be a set of Boolean variables, which take on the values $0$ (false) and $1$ (true).  If $x$ is a variable in $U$, then $x$ and $\overline{x}$ are called literals, with the value of $\overline{x}$ defined to be $1$ if $x$ is $0$ and vice versa.  A clause is any set of literals over $U$, such as $\{\overline{x}_1, x_3, x_4\}$ which represents ``$x_1$ is false or $x_3$ is true or $x_4$ is true'', and the clause is satisfied by a truth assignment whenever at least one of its literals is true under that assignment.

An instance of PLANAR 3SAT is a collection $C$ of clauses over $U$, each containing either $2$ or $3$ literals, such that the following bipartite graph is planar: the vertex set is $U \cup C$, and the pair $\{x, c\}$ is an edge whenever either literal $x$ or $\overline{x}$ belongs to the clause $c$.  The PLANAR 3SAT problem asks whether there exists some truth assignment for $U$ that simultaneously satisfies all the clauses in $C$.  This problem is known to be NP-hard.  It is even known to be NP-hard under the additional hypothesis that every variable $x$ occurs exactly once positively (i.e., as $x$) and exactly twice negatively (i.e., as $\overline{x}$) in $C$; this problem is called $\psat$, which stands for PLANAR 3SAT WITH EXACTLY 3 OCCURRENCES PER VARIABLE~\cite{Cerioli08}.

To prove Theorem~\ref{teo:nphard-rook}, we give a polynomial reduction that constructs a rook-visibility guard set problem for each instance of \psat.  Given an instance of $\psat$, we replace each variable by a \textit{\textbf{variable gadget}}, each clause by a \textit{\textbf{clause gadget}}, and each edge between variable and clause by a \textit{\textbf{connection gadget}}.  

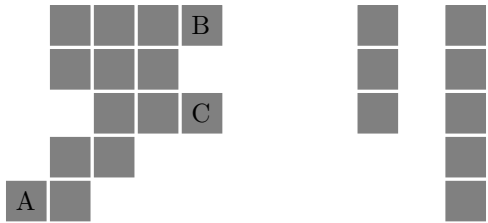
\begin{figure}
\begin{center}
\begin{tikzpicture} 
\foreach \x/\y in { 1/0, 2/0, 2/1, 2/3, 2/4, 3/1, 3/2, 3/3, 3/4, 4/2 , 4/3, 4/4, 5/2 ,  5/4, 9/2, 9/3, 9/4, 11/0, 11/1, 11/2, 11/3, 11/4} { 
\path [draw=gray, fill=gray] (.5+\x-0.45, .5+\y-0.45)
-- ++(0,.9)
-- ++(.9,0)
-- ++(0,-.9)
--cycle;
}

\node[anchor=west] at (1+.05, 0+.5){A};
\node[anchor=west] at (5+.05, 4+.5){B};
\node[anchor=west] at (5+.05, 2+.5){C};
\end{tikzpicture}
\end{center}
\caption{In the polyomino corresponding to an instance of $\psat$, each variable is represented by the variable gadget shown on the left, and each clause is represented by one of the clause gadgets shown on the right, according to whether the clause contains two literals or three.}\label{fig:gadgets}
\end{figure}

Each connection gadget is a path one tile wide from a variable gadget to a clause gadget.  The variable gadget appears in the left of Figure~\ref{fig:gadgets}.  Recall that each variable appears in one clause positively and in two clauses negatively.  At the variable gadget, the connection gadget to the clause with the positive literal extends out to the left from tile $A$, and the connection gadgets to the two clauses with the negative literals extend out to the right from tiles $B$ and $C$.  The clause gadget is either three vertically-consecutive tiles, if the clause contains two literals, or five vertically-consecutive tiles, if the clause contains three literals.  These two options appear in the right of Figure~\ref{fig:gadgets}.  The connection gadgets extend sideways from the first and third tiles, or from the first, third, and fifth tiles.

The proof of Lemma~\ref{lem:polynomial-size} gives more detail on how to construct the polyomino from the three types of gadgets.  First, we prove the distinctive properties of the variable gadget that will be needed in the main proof to encode $\psat$.

\begin{lemma}\label{lem:variable-gadget}
There is only one way to guard the variable gadget with four rooks if one of them is at $A$, and there is only one way to guard the variable gadget with four rooks if two of them are at $B$ and $C$. 
\end{lemma}




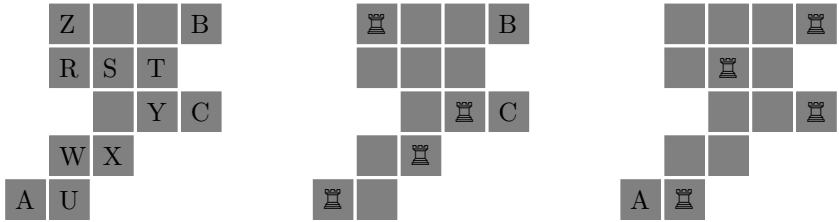
\begin{figure}[h]
\begin{center}
\begin{tikzpicture} 
\foreach \x/\y in { 1/0, 2/0, 2/1, 2/3, 2/4, 3/1, 3/2, 3/3, 3/4, 4/2 , 4/3, 4/4, 5/2 ,  5/4} { 
\path [draw=gray, fill=gray] (.5+\x-0.45, .5+\y-0.45)
-- ++(0,.9)
-- ++(.9,0)
-- ++(0,-.9)
--cycle;
}

\node[anchor=west] at (0+1.05, 0+.5){A};
\node[anchor=west] at (5+.05, 4+.5){B};
\node[anchor=west] at (5+.05, 2+.5){C};
\node[anchor=west] at (2+.05, 3+.5){R};
\node[anchor=west] at (3+.05, 3+.5){S};
\node[anchor=west] at (4+.05, 3+.5){T};
\node[anchor=west] at (2+.05, 0+.5){U};
\node[anchor=west] at (3.05, 1+.5) {X};
\node[anchor=west] at (4.05, 2+.5) {Y};
\node[anchor=west] at (2.05, 4+.5) {Z};
\node[anchor=west] at (+0+2.05, 0+1.5){W};

\foreach \x/\y in { 1/0, 2/0, 2/1, 2/3, 2/4, 3/1, 3/2, 3/3, 3/4, 4/2 , 4/3, 4/4, 5/2 ,  5/4} { 
\path [draw=gray, fill=gray] (7.5+\x-0.45, .5+\y-0.45)
-- ++(0,.9)
-- ++(.9,0)
-- ++(0,-.9)
--cycle;
}

\node[anchor=west] at (7+1.05, 0+.5) {\rook};
\node[anchor=west] at (7+3.05, 1+.5) {\rook};
\node[anchor=west] at (7+4.05, 2+.5) {\rook};
\node[anchor=west] at (7+2.05, 4+.5) {\rook};

\node[anchor=west] at (7+5+.05, 4+.5){B};
\node[anchor=west] at (7+5+.05, 2+.5){C};

\foreach \x/\y in { 1/0, 2/0, 2/1, 2/3, 2/4, 3/1, 3/2, 3/3, 3/4, 4/2 , 4/3, 4/4, 5/2 ,  5/4} { 
\path [draw=gray, fill=gray] (14.5+\x-0.45, .5+\y-0.45)
-- ++(0,.9)
-- ++(.9,0)
-- ++(0,-.9)
--cycle;
}
\node[anchor=west] at (7+7+2.05, 0+.5) {\rook};
\node[anchor=west] at (7+7+3.05, 3+.5) {\rook};
\node[anchor=west] at (7+7+5.05, 4+.5) {\rook};
\node[anchor=west] at (7+7+5.05, 2+.5) {\rook};
\node[anchor=west] at (14+0+1.05, 0+.5){A};

\end{tikzpicture}
\end{center}
\caption{Labeled variable gadget and unique four-guards sets.}
\label{fig:coveringwithfour}
\end{figure}
\begin{proof}
In Figure \ref{fig:coveringwithfour} we depict the variable gadget and we label some of its tiles. Also, we depict the unique configurations stated in this lemma.

If a rook is on $A$, then the remaining three rooks have to guard tiles $W$, $B$, and $C$. Observe that no two of these tiles can be guarded by the same rook because they are in different rows and in different columns. Then three rooks are needed to guard tiles $W$, $B$, and $C$, and this cannot be accomplished if one of these three rooks is in the row containing tiles $R$, $S$, and $T$, because these tiles do not share rows or columns in common with tiles $W$, $B$, or $C$. Also, the rook on $A$ cannot guard tiles $R$, $S$, or $T$. This forces the rooks guarding $W$, $B$, and $C$ to also guard tiles $R$, $S$, and $T$. This can only be done if the rooks are placed on tiles $X$, $Y$, and $Z$.    
 
If two rooks are on tiles $B$ and $C$, then there is a third different rook that has to guard tile $A$. None of these three rooks are able to guard tiles $R$, $S$, or $T$. Thus a fourth rook has to be placed in this row. Tile $W$ cannot be guarded by the rooks at $B$ and $C$, or by the rook that guards $R$, $S$, and $T$.  Thus, $W$ has to be guarded by the rook guarding $A$, and this forces that third rook to be placed on $U$. Also, tile $X$ has to be guarded, and it cannot be guarded by the rooks guarding tiles $B$, $C$, or $A$. This forces the fourth rook to be placed on $S$.

\end{proof}

\begin{corollary}\label{cor}
Let $G$ be an arbitrary four-rook set guarding all tiles of the variable gadget. If $G$ has a rook on tile $B$ or $C$, then $G$ has no rook on tile $A$, and if $G$ has a rook on $A$ then it has no rooks on $B$ or $C$. Also, there does not exist a three-rook set guarding all tiles of the variable gadget.
\end{corollary}
\begin{proof}
In Lemma \ref{lem:variable-gadget} we have proved that in any four-rook set guarding the variable gadget, once a rook is placed on $A$ then there cannot be a rook placed on $B$ or $C$. 
If there exists a three-rook set guarding all tiles of the variable gadget, then none of the rooks is on $A$ because of the uniqueness proved in Lemma \ref{lem:variable-gadget}. Then we can place a fourth rook on $A$. This gives us a four-rook set guarding the variable gadget with a rook placed on $A$; thus this set has to be the one described in Lemma \ref{lem:variable-gadget}, and then if the rook placed at $A$ is removed, then tile $A$ is not guarded. Thus there cannot exist a three-rook set guarding the variable gadget. 
\end{proof}

In the next lemma, we check that given an instance of $\psat$, the corresponding instance of the rook-visibility guard set problem is not too large.

\begin{lemma}\label{lem:polynomial-size}
Let $C$ be an instance of $\psat$ with $n$ variables.  We can construct a polyomino $P(C)$ with $O(n^2)$ tiles, formed by taking a planar drawing of the instance $C$ and replacing the variable nodes by variable gadgets, the clause nodes by clause gadgets, and the edges by connection gadgets.
\end{lemma}

\begin{proof}
If we did not care about the size of the polyomino, we would simply take any smoothly differentiable planar drawing of $C$ and scale it up so large that when we superimposed it on the grid of unit squares, we could locally replace the vertices by variable and clause gadgets and the edges by grid paths.  The question is how to limit the amount of scaling up that we need to do in order to make these local replacements.

We use a procedure described in~\cite{Biedl13} that creates a planar drawing of the graph of $C$ on the grid of unit squares in such a way that every vertex is a unit square, every edge is a path of unit squares between the two vertices of the edge, and the whole drawing fits in an $O(n)$ by $O(n)$ square.  For convenience we describe the process again here.  Figure~\ref{fig:vis} depicts the process where the graph is the cube.

A (weak) \textbf{\textit{visibility representation}} of a planar graph is a way to draw it in the plane such that each vertex is a horizontal segment, each edge is a vertical segment, the endpoints of each edge are on the horizontal segments corresponding to the two incident vertices, and there are no other intersections among the segments.  Every planar graph has a visibility representation.  A quick proof of this fact, based on the fact that every planar graph can be drawn with straight edges, can be found in~\cite{Duchet83}.

Given a visibility representation of our planar graph, we look at the collection of $x$-coordinates and $y$-coordinates of the endpoints of all the segments, and assume without loss of generality that the only repeated coordinates are those that are implied by having a visibility representation.  We thicken each segment slightly to get a collection of rectangles with the same combinatorial relationship as in our desired end result.  Then we make a new set of rectangles with the same combinatorial relationship, by moving all the horizontal coordinates to consecutive integer values while keeping them in the same order, and similarly for the vertical coordinates.  

In our instance $C$ of $\psat$, there are $n$ variables, so there are exactly $3n$ edges and fewer than $6n$ total nodes.  Our drawing has one rectangle for each node and one rectangle for each edge, so there are at most $9n$ rectangles, with at most $18n$ horizontal coordinates and at most $18n$ vertical coordinates.  Thus our new drawing fits in an $18n$ by $18n$ square.

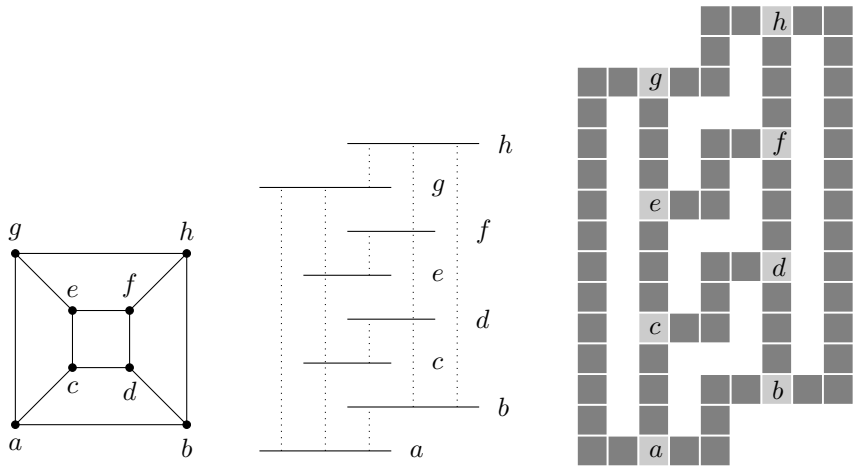
\begin{figure}
\begin{center}
\begin{tikzpicture}[scale=1.3, vert/.style={circle, draw=black, fill=black, inner sep = 0pt, minimum size = 1mm}]
\node[vert] at (1, 1) [label=below: $a$] {};
\node[vert] at (4, 1) [label=below: $b$] {};
\node[vert] at (2, 2) [label=below: $c$] {};
\node[vert] at (3, 2) [label=below: $d$] {};
\node[vert] at (2, 3) [label=above: $e$] {};
\node[vert] at (3, 3) [label=above: $f$] {};
\node[vert] at (1, 4) [label=above: $g$] {};
\node[vert] at (4, 4) [label=above: $h$] {};
\draw (1, 1)--(4, 1)--(4, 4)--(1, 4)--(1, 1) (2, 2)--(3, 2)--(3, 3)--(2, 3)--(2, 2) (1, 1)--(2, 2) (4, 1)--(3, 2) (4, 4)--(3, 3) (1, 4)--(2, 3);
\end{tikzpicture} \hspace{.2in}
\begin{tikzpicture}
\draw (.5, 1)--(3.5, 1) (2.5, 2)--(5.5, 2) (1.5, 3)--(3.5, 3) (2.5, 4)--(4.5, 4) (1.5, 5)--(3.5, 5) (2.5, 6)--(4.5, 6) (.5, 7)--(3.5, 7) (2.5, 8)--(5.5, 8);
\draw[dotted] (1, 1)--(1, 7) (2, 1)--(2, 7) (3, 1)--(3, 2) (4, 2)--(4, 8) (5, 2)--(5, 8) (3, 3)--(3, 4) (3, 5)--(3, 6) (3, 7)--(3, 8);
\node at (3.5, 1) [label=right: $a$] {};
\node at (5.5, 2) [label=right: $b$] {};
\node at (4, 3) [label=right: $c$] {};
\node at (5, 4) [label=right: $d$] {};
\node at (4, 5) [label=right: $e$] {};
\node at (5, 6) [label=right: $f$] {};
\node at (4, 7) [label=right: $g$] {};
\node at (5.5, 8) [label=right: $h$] {};
\end{tikzpicture} \hspace{.2in}
\begin{tikzpicture}[scale=.7]
\foreach \x/\y in { 1/1, 2/1, 4/1, 5/1, 1/2, 3/2, 5/2, 1/3, 3/3, 5/3, 6/3, 8/3, 9/3, 1/4, 3/4, 7/4, 9/4, 1/5, 4/5, 5/5, 7/5, 9/5, 1/6, 3/6, 5/6, 7/6, 9/6, 1/7, 3/7, 5/7, 6/7, 9/7, 1/8, 3/8, 7/8, 9/8, 1/9, 4/9, 5/9, 7/9, 9/9, 1/10, 3/10, 5/10, 7/10, 9/10, 1/11, 3/11, 5/11, 6/11, 9/11, 1/12, 3/12, 7/12, 9/12, 1/13, 2/13, 4/13, 5/13, 7/13, 9/13, 5/14, 7/14, 9/14, 5/15, 6/15, 8/15, 9/15} { 
\path [draw=gray, fill=gray] (.5+\x-0.45, .5+\y-0.45)
-- ++(0,.9)
-- ++(.9,0)
-- ++(0,-.9)
--cycle;
}

\foreach \x/\y in {3/1, 7/3, 3/5, 7/7, 3/9, 7/11, 3/13, 7/15} { 
\path [draw=gray!40, fill=gray!40] (.5+\x-0.45, .5+\y-0.45)
-- ++(0,.9)
-- ++(.9,0)
-- ++(0,-.9)
--cycle;
}

\node[anchor=west] at (3+.05, 1+.5){$a$};
\node[anchor=west] at (7+.05, 3+.5){$b$};
\node[anchor=west] at (3+.05, 5+.5){$c$};
\node[anchor=west] at (7+.05, 7+.5){$d$};
\node[anchor=west] at (3+.05, 9+.5){$e$};
\node[anchor=west] at (7+.05, 11+.5){$f$};
\node[anchor=west] at (3+.05, 13+.5){$g$};
\node[anchor=west] at (7+.05, 15+.5){$h$};
\end{tikzpicture}
\end{center}
\caption{Using a visibility representation, any planar graph can be drawn on the square grid inside a box of side length $O(\abs{V} + \abs{E})$.}\label{fig:vis}
\end{figure}

At this stage, each node of the graph of $C$ is represented by a horizontal strip of squares.  But, we know that every node has degree $2$ or $3$.  Thus, for each node, if it has degree $2$ we select any square of the strip to be the node square, and think of the remaining squares as edge squares.  If it has degree $3$ we select the node square to be the square of the strip that touches the middle one of the three incident edges, and think of the remaining squares as parts of the two outer edges.  In this way, we have succeeded in drawing the graph of $C$ on the square grid in a square of side length $O(n)$.

The only further modification we need is to accommodate the fact that the variable gadget and the clause gadgets are more than one square large.  To do this, we subdivide the drawing we have made---replacing each unit square by a $15$ by $15$ grid should be enough.  This is enough space for the variable gadgets and the clause gadgets to be placed in the middle of the corresponding node squares, and for the connection gadgets to extend outward from them toward the corresponding edge squares so that they can follow the middles of the edge squares, as in Figure~\ref{fig:tangle}.  The resulting polyomino $P(C)$ has $O(n^2)$ tiles.
\end{proof}

\begin{figure}
\begin{center}
\begin{tikzpicture}[scale=.7]
\foreach \x in {-7, -6, -5, -4, -3, -2, -1, 0, 1, 2, 3, 4, 5, 6, 7}{
\foreach \y in {-7, -6, -5, -4, -3, -2, -1, 0, 1, 2, 3, 4, 5, 6, 7}{
\path [draw=gray!20, fill=gray!20] (.5+\x-0.45, .5+\y-0.45)
-- ++(0,.9)
-- ++(.9,0)
-- ++(0,-.9)
--cycle;
}
}
\foreach \x/\y in {-4/-2, -3/-2, -3/-1, -2/-1, -2/0, -1/0, 0/0, -3/1, -2/1, -1/1, -3/2, -2/2, -1/2, 0/2} { 
\path [draw=gray!60, fill=gray!60] (.5+\x-0.45, .5+\y-0.45)
-- ++(0,.9)
-- ++(.9,0)
-- ++(0,-.9)
--cycle;
}
\foreach \x/\y in {-5/-2, -6/-2, -6/-3, -6/-4, -5/-4, -4/-4, -3/-4, -2/-4, -1/-4, 0/-4, 1/-4, 2/-4, 3/-4, 4/-4, 5/-4, 6/-4, 6/-3, 6/-2, 6/-1, 6/0, 7/0, 1/0, 2/0, 3/0, 4/0, 4/1, 4/2, 4/3, 4/4, 4/5, 4/6, 3/6, 2/6, 1/6, 0/6, 0/7, 1/2, 2/2, 2/3, 2/4, 1/4, 0/4, -1/4, -2/4, -3/4, -4/4, -5/4, -5/3, -5/2, -5/1, -5/0, -6/0, -7/0} { 
\path [draw=gray, fill=gray] (.5+\x-0.45, .5+\y-0.45)
-- ++(0,.9)
-- ++(.9,0)
-- ++(0,-.9)
--cycle;
}
\end{tikzpicture}
\end{center}
\caption{Within each subdivided vertex square, the connection gadgets may have to wind around to get from their starting points on the variable or clause gadget to their neighboring edge squares.}\label{fig:tangle}
\end{figure}

The next lemma statement involves counting bends of connection gadgets.  A \textit{\textbf{bend}} is a tile where the path of the connection gadget changes between horizontal and vertical.  Because each connection gadget is horizontal on both ends where it connects to the variable and clause gadgets, it must have an even number of bends.  The next two lemmas show that given an instance of $\psat$, the corresponding instance of the rook-visibility guard set problem is indeed equivalent.

\begin{lemma}\label{lem:can-place-rooks}
Let $C$ be an instance of $\psat$ as in Lemma~\ref{lem:polynomial-size}, and let $P(C)$ be its associated polyomino.  Let $m = 4n + \frac{\ell}{2}$, where $n$ is the number of variables and $\ell$ is the total number of bends in all the connection gadgets.  If $C$ is satisfiable, then $P(C)$ can be guarded by $m$ rooks.
\end{lemma}

\begin{proof}
Consider a truth assignment that satisfies all the clauses in $C$.  For each variable, if it is true in this truth assignment, we guard the corresponding variable gadget with the four-rook set from Lemma~\ref{lem:variable-gadget} with a rook on tile $A$; if it is false in this truth assignment, we use the set with rooks on $B$ and $C$.

For each edge, we say that the edge joining a variable with a clause is \textbf{\textit{hot}} if the variable is true and it appears in the clause, or if the variable is false and its negation appears on the clause.  So, for each hot edge, we have placed a rook on $A$, $B$, or $C$ of the variable gadget; move that rook sideways along the corresponding connection gadget to land on the first bend (or all the way to the clause gadget if there is no bend).  Then place rooks at every second bend for the rest of the hot edge, ending with one rook on the clause gadget.  For each non-hot edge, we start with a rook at the second bend leading away from the variable gadget (the first bend is already guarded by a rook on the variable gadget), and place rooks at every second bend for the rest of the connection gadget, ending with one rook at the last bend before the clause gadget.

We have placed $4$ rooks on each variable gadget, and on each connection gadget we have placed a number of additional rooks equal to half the number of bends.  In this way the variable gadgets and connection gadgets are guarded.  We also see that the clause gadgets are guarded because the
truth assignment satisfies all the clauses, so each clause has a hot edge, and there is a rook on the clause gadget at the end of that hot edge.
\end{proof}

\begin{lemma}\label{lem:is-satisfiable}
Let $C$, $P(C)$, and $m$ be as in Lemma~\ref{lem:can-place-rooks}.  If the polyomino $P(C)$ can be guarded by $m$ rooks, then the $\psat$  instance $C$ is satisfiable.
\end{lemma}

\begin{proof}
Suppose we have a set of $m$ rooks that guards $P(C)$.  Then each clause gadget has at least one guard.  We can assume that it is on one of the squares at the end of a connection gadget (if not, we can move it and everything is still guarded).  Similar to the proof of Lemma~\ref{lem:can-place-rooks}, we call those edges the hot edges.  We will deduce the truth assignment for which these edges have the same property that defined the hot edges of Lemma~\ref{lem:can-place-rooks}.

We make the following observation.  Given a connection gadget with $b$ bends, if one end is already guarded up through the first bend, then we can guard the remainder with $\frac{b}{2}$ rooks, but if so, there cannot be a rook past the last bend.  To prove this, we can move every rook that appears between bends to the next bend, and then observe that at least every second bend needs a rook.

Next we apply this connection observation to each edge.  For each hot edge, we think of the first end as being the one at the clause gadget, which is already guarded by that first rook.  The connection observation implies that there must be at least $\frac{b}{2}$ more rooks on that connection gadget.  For each edge that is not hot, we think of the first end as being the one at the variable gadget, which may already be guarded; the connection observation implies that there must be at least $\frac{b}{2}$ rooks on the connection gadget itself.

Corollary~\ref{cor} implies that there must be at least four rooks on each variable gadget or up through the first bend of each connection gadget coming out of it.  For each hot edge, we need one guard to be shared by the variable gadget and the connection gadget, and this guard has to be at the first bend of the connection gadget, otherwise the connection observation implies that we do not have enough guards.  We say the variable is true if it has a hot edge coming into the $A$ side, and say the variable is false if it has a hot edge coming into the $B$ or $C$ side.  Then Corollary~\ref{cor} implies that the variable cannot be both true and false if the gadget is guarded by only four rooks.  Thus, with $m = 4n + \frac{\ell}{2}$ rooks in total, we have accounted for all the rooks, and so we have found a well-defined truth assignment that satisfies our $\psat$ instance $C$.
\end{proof}

\begin{proof}[Proof of Theorem~\ref{teo:nphard-rook}]
The polyomino $P(C)$ from Lemma~\ref{lem:polynomial-size} not only is of polynomial size, but in fact can be computed in polynomial time, because finding a visibility representation of a planar graph can be done in linear time~\cite{Tamassia86}.  Thus there is a polynomial-time algorithm to convert any instance $C$ of $\psat$ into a polyomino $P(C)$ and a number $m = 4n + \frac{\ell}{2}$, for which Lemmas~\ref{lem:can-place-rooks} and~\ref{lem:is-satisfiable} imply that the rook-visibility guard set problem for $P(C)$ and $m$ is equivalent to the $\psat$ problem for $C$.
The $\psat$ problem is NP-hard~\cite{Cerioli08}, and thus the rook-visibility guard set problem must also be NP-hard.
\end{proof}

The proof of Theorem~\ref{teo:nphard-queen} about queens is similar, but the variable gadget is different and we need to take a little bit more care with the connection gadgets.  First we give the queen-variable gadget, shown in Figure~\ref{fig:coveringwiththree}, and show that it has the properties we need for the rest of the proof. 

\begin{figure}[h]
\begin{center}
\begin{tikzpicture} 
\foreach \x/\y in {0/2, 1/2, 2/2, 2/0, 2/1,  2/3, 2/4, 3/1, 3/3, 4/0 , 4/1, 4/3, 4/4, 5/1 ,  5/3, 6/1, 6/3} { 
\path [draw=gray, fill=gray] (-2+.5+\x-0.45, .5+\y-0.45)
-- ++(0,.9)
-- ++(.9,0)
-- ++(0,-.9)
--cycle;
}

\node[anchor=west] at (-2+0.05, 2+.5){A};
\node[anchor=west] at (0.05, 2+.5){Q};
\node[anchor=west] at (0.05, 4+.5){X};
\node[anchor=west] at (2+0.05, 4+.5){Y};
\node[anchor=west] at (1+0.05, 3+.5){R};
\node[anchor=west] at (2+0.05, 3+.5){S};
\node[anchor=west] at (3+0.05, 3+.5){T};
\node[anchor=west] at (4+0.05, 3+.5){B};
\node[anchor=west] at (4+0.05, 1+.5){C};
\node[anchor=west] at (0.05, 0+.5){Z};
\node[anchor=west] at (2+0.05, 0+.5){W};
\node[anchor=west] at (1.05, 1+.5){L};
\node[anchor=west] at (2+0.05, 1+.5){U};
\node[anchor=west] at (3+0.05, 1+.5){V};
\end{tikzpicture}
\begin{tikzpicture}

\foreach \x/\y in { 0/2, 1/2, 2/2, 2/0, 2/1,  2/3, 2/4, 3/1, 3/3, 4/0 , 4/1, 4/3, 4/4, 5/1 ,  5/3, 6/1, 6/3} { 
\path [draw=gray, fill=gray] (6.5+\x-0.45, .5+\y-0.45)
-- ++(0,.9)
-- ++(.9,0)
-- ++(0,-.9)
--cycle;
}

\node[anchor=west] at (5+1.05, 2+.5) {\queen};
\node[anchor=west] at (6+3.05, 3+.5) {\queen};
\node[anchor=west] at (6+3.05, 1+.5) {\queen};

\node[anchor=west] at (7+5+.05, 3+.5){B};
\node[anchor=west] at (7+5+.05, 1+.5){C};

\foreach \x/\y in {0/2, 1/2, 2/2, 2/0, 2/1,  2/3, 2/4, 3/1, 3/3, 4/0 , 4/1, 4/3, 4/4, 5/1 ,  5/3, 6/1, 6/3} { 
\path [draw=gray, fill=gray] (14.5+\x-0.45, .5+\y-0.45)
-- ++(0,.9)
-- ++(.9,0)
-- ++(0,-.9)
--cycle;
}
\node[anchor=west] at (10+5+1.05, 2+.5) {\queen};
\node[anchor=west] at (11+6+3.05, 3+.5) {\queen};
\node[anchor=west] at (11+6+3.05, 1+.5) {\queen};
\node[anchor=west] at (13+0+1.05, 2+.5){A};

\end{tikzpicture}
\end{center}
\caption{Labeled queen-variable gadget and unique three-queen-guards sets.}
\label{fig:coveringwiththree}
\end{figure}
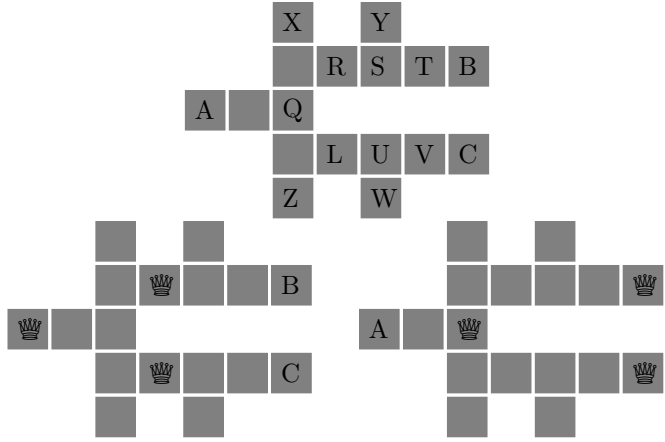

\begin{lemma}\label{lem:queen-variable}
The queen-variable gadget cannot be guarded by two queens.  It can be guarded by three queens, of which one is at $A$, or by three queens, of which two are at $B$ and $C$.  It cannot be guarded by three queens, of which two are at $A$ and $B$ or at $A$ and $C$.
\end{lemma}

\begin{proof}
That at least three queens are needed to guard the gadget polyomino follows from the fact that no two of the tiles $A$, $B$, or $C$ can be guarded by the same queen.

We show that there is a unique way to guard the polyomino gadget with three queens if a queen is placed at $A$. If a queen is placed at $A$, then none of the tiles $C$, $B$, $X$, $Y$, $Z$, and $W$ are guarded by this queen. First, we notice that to guard tiles $C$ and $B$ one queen must be placed at some tile along the row containing $B$ and another queen must be placed at some tile along the row containing $C$. Also, it is not possible to guard both $Y$ and $W$ with the same queen unless this queen is at $Q$, but if a queen is placed at $Q$ then it will not be in either of the rows containing $B$ or $C$. Thus, to guard $Y$ a queen must be placed at either $R$, $S$, or $T$, and to guard $W$ a queen must be placed at either $L$, $U$, or $V$. The only options that also allow to cover tiles $X$ and $Z$ is to place these two queens at $R$ and $L$. This gives the unique way to guard the gadget polyomino with three queens if one is placed at $A$.

Now, we prove that there is a unique way to guard the polyomino gadget with three queens if two of these queens are placed at $B$ and $C$.  If two queens are placed at $B$ and $C$, then tiles $X$, $A$, and $Z$ are not guarded. The only way to guard these tiles with one more queen is to place it at tile $Q$, and this position of the third queen guards the rest of the gadget polyomino tiles. This gives the unique way of guarding the gadget polyomino if two queens are placed at $B$ and $C$.

Finally, as a consequence of the uniqueness of these two three-queen-guard sets, there is no way to guard the polyomino gadget with three queens in such a way that a queen is at $A$ and another queen is either at $B$ or $C$.
\end{proof}

\begin{proof}[Proof of Theorem~\ref{teo:nphard-queen}]
Lemma~\ref{lem:polynomial-size}, about how to construct the polyomino associated to an instance of $\psat$, applies to queens if we use the queen-variable gadget.  We also need to make sure that for each connection gadget, there must be at least three tiles in each segment between bends and in each segment between the last bend and a variable or clause gadget.  This is not a concern in the subdivided edge squares, because there is already plenty of space between bends, and only applies to the subdivided node squares---we may need more than a $15$ by $15$ subdivision of all squares.

Then Lemmas~\ref{lem:can-place-rooks} and~\ref{lem:is-satisfiable} apply to queens, if we replace $m = 4n + \frac{\ell}{2}$ by $m = 3n + \frac{\ell}{2}$ to account for the fact that the queen-variable gadget needs three queens and not four.  In Lemma~\ref{lem:is-satisfiable}, we use the spacing of the bends in the connection gadgets to ensure that a segment between two bends cannot be guarded entirely by queens beyond those two bends.

Using these lemmas as modified for the queen case, we have shown that for every instance of $\psat$, we can compute in polynomial time a polynomial-size instance of the queen-visibility guard set problem that is equivalent.  Thus the queen-visibility guard set problem is NP-hard.
\end{proof}

\section{Non-Attacking Rooks on Polyominoes}\label{sec:nonattacking}

In this section we study the special case of rooks on 2--dimensional polyominoes.  In this case we have the convenient tool of being able to encode the attacking relationships between rooks as a bipartite graph.  We first construct this bipartite graph, and then in the rest of the section we give some theorems in which we use the properties of the bipartite graph to deduce properties of the rook placements on polyominoes.

In the bipartite graph, the two partite sets are the set of rows and the set of columns of the polyomino.  To be precise, we say that two tiles $X$ and $Y$ in a polyomino $P$ are in the same \textbf{\textit{row}} if the line connecting their centers is horizontal and contained in $P$. Similarly, two tiles $X$ and $Y$ of a polyomino $P$ are in the same \textbf{\textit{column}} if the line connecting their centers is vertical and contained in $P$.  This definition of row and column may be counter-intuitive, in that two tiles of $P$ may be in the same row or column of the infinite chessboard without being in the same row or column of $P$, but it describes the possible lines of attack of rooks on the polyomino. 
 
Given a polyomino $P$ we denote the corresponding bipartite graph by $G(P)=(V,E)$.  The set of vertices is given by $V=C \cup R$ where $C$ has a vertex for each column and $R$ has a vertex for each row, and the graph includes an edge $(c,r)=(r,c)$ if there is a tile at column $c$ and row $r$. An example of this construction appears in Figure~\ref{fig:examplebipartitegraph}.

\begin{figure}
\begin{center}
\begin{tikzpicture}[scale=1, vert/.style={circle, draw=black, fill=black, inner sep = 0pt, minimum size = 1mm}]
\foreach \x/\y in {  0 / 2, 1 / 2,   2 / 2, 0 / 4, 0 / 3, 1 / 4, 2 / 4, 2 / 3 } { 
\path [draw=gray!, fill=gray!] (-8+3.5+\x-0.45, .5+\y-1.45)
-- ++(0,.9)
-- ++(.9,0)
-- ++(0,-.9)
--cycle;
}

\node[anchor=west] at (-6+.15, 3+.5){$r_1$};
\node[anchor=west] at (-6+.15, 2+.5){$r_2$};
\node[anchor=west] at (-2+.035, 2+.5){$r_3$};
\node[anchor=west] at (-6+.15, 1+.5){$r_4$};

\node[anchor=west] at (-5+.05, 4+.3){$c_1$};
\node[anchor=west] at (-4+.05, 4+.3){$c_2$};
\node[anchor=west] at (-4+.05, .6){$c_3$};
\node[anchor=west] at (-3+.05, 4+.3){$c_4$};

\foreach \x/\y in {  0 / 2, 1 /1, 1/2, 1/3,   2 / 0, 2 / 1, 2/2, 3/1} { 
\path [draw=gray!, fill=gray!] (7.5+\x-0.45-1, 1.5+\y-1.45)
-- ++(0,.9)
-- ++(.9,0)
-- ++(0,-.9)
--cycle;
}

\node[anchor=west] at (6+.15-1, 3+.5){$r_2$};
\node[anchor=west] at (6+.15-1, 2+.5){$r_1$};
\node[anchor=west] at (6+.15-1, 0+.5){$r_3$};
\node[anchor=west] at (6+.15-1, 1+.5){$r_4$};

\node[anchor=west] at (7+.05-1, 4+.3){$c_2$};
\node[anchor=west] at (8+.05-1, 4+.3){$c_1$};
\node[anchor=west] at (9+.05-1, 4+.3){$c_4$};
\node[anchor=west] at (10+.05-1, 4+.3){$c_3$};

\node[vert] at (1, 1) [label=left: $r_1$] {};
\node[vert] at (3, 1) [label=right: $c_1$] {};
\node[vert] at (1, 2) [label=left: $r_2$] {};
\node[vert] at (3, 2) [label=right: $c_2$] {};
\node[vert] at (1, 3) [label=left: $r_3$] {};
\node[vert] at (3, 3) [label=right: $c_3$] {};
\node[vert] at (1, 4) [label=left: $r_4$] {};
\node[vert] at (3, 4) [label=right: $c_4$] {};

\draw (1, 1)--(3, 1) (1, 1)--(3, 2) (1, 1)--(3, 4) (1, 2)--(3, 1) (1, 3)--(3, 4) (1, 4)--(3, 4) (1, 4)--(3, 3) (1, 4)--(3, 1);
\end{tikzpicture}
\caption{Given a polyomino $P$, the corresponding bipartite graph $G(P)$ has one vertex for each row and column, and one edge for each tile.  The two polyominoes shown both give the same bipartite graph.}\label{fig:examplebipartitegraph}
\end{center}
\end{figure}
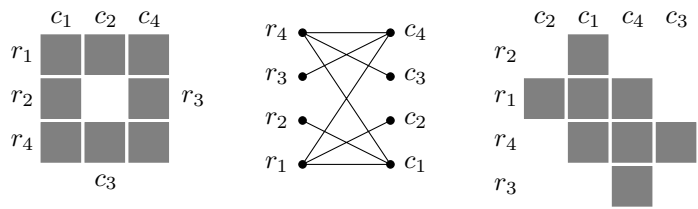

This bipartite graph is analogous to the one for grid art galleries in~\cite{Ntafos86}, and a similar, but not identical, bipartite graph construction appears in the paper~\cite{Dailly2018}, which studies the two-player game of placing non-attacking rooks one at a time until the rooks guard the entire polyomino. Another object that records non-attacking rook positions is the chessboard complex, which has one vertex for each tile, and connects sets of tiles that constitute non-attacking positions~\cite{bjorner1994chessboard}. It is possible to construct the analogue of the chessboard complex using non-attacking positions for other chess pieces. As far as we know, this is the first time that this chessboard complex is studied in polyominoes that are not rectangular polyominoes. For the purposes of this paper, the bipartite graph $G(P)$ is the most useful model but the results can be interpreted in terms of the chessboard complex.

The information from the bipartite graph tells us many properties of the polyomino, but studying polyominoes is not identical to studying bipartite graphs.  In Figure \ref{fig:examplebipartitegraph} we show two polyominoes that are not congruent or even topologically equivalent, but their bipartite graphs are isomorphic. And, not all bipartite graphs correspond to polyominoes.  For the bipartite graph $C_6$, the cycle with six vertices, there does not exists a polyomino $P$ such that $G(P)=C_6$. 





The bipartite graph construction is especially useful for studying sets of non-attacking rooks, which correspond to matchings in the graph.  A set of rooks is \textbf{\textit{non-attacking}} if no two rooks are in the same row or column (in the sense defined above), and a \textbf{\textit{matching}} in any graph is a set of edges that do not share any vertices.  Thus, efficient algorithms for finding maximum matchings in bipartite graphs also give maximum non-attacking rook sets in polyominoes.


 
\begin{theorem}\label{teo:polynomial}
There is a polynomial time algorithm to find the maximum number of non-attacking rooks on a polyomino.
\end{theorem} 

\begin{proof}
 The maximum number of non-attacking rooks on a polyomino $P$ is equal to the maximum number of edges in a matching in the bipartite graph $G(P)$.  A maximum matching in a bipartite graph can be found in polynomial time using an augmenting-path algorithm.
\end{proof}

The theorem above can be strengthened to address a rook analogue of the problem from~\cite{gent2017complexity}, which says that it is NP-hard to decide whether a given set of non-attacking queens can be extended.  In the rook case, the problem is polynomial-time solvable, even on polyominoes.

\begin{theorem}\label{teo:extend}
There is a polynomial time algorithm to decide for any set of $m$ non-attacking rooks on a polyomino and any $k>0$, if it can be extended to a set of $m+k$ non-attacking rooks.
\end{theorem} 

\begin{proof}
It suffices to show how to decide, for any matching of size $m$ in a bipartite graph $G$ and any $k>0$, whether it can be extended to a matching of size $m + k$.  Let $M$ be the matching of size $m$.  Let $G'$ be the graph obtained by deleting $M$, all vertices in $M$, and all edges incident to those vertices, from $G$.  Then $G$ has a matching of size $m+k$ that contains $M$ if and only if $G'$ has a matching of size $k$.  Thus we may solve the problem by running an augmenting-path algorithm on $G'$.
\end{proof}

These theorems show that finding maximum-size sets of non-attacking rooks is computationally easy.  What about if instead we consider the possible sizes of sets of non-attacking rooks that are maximal, in the sense that no more rooks can be added without destroying the non-attacking property?  Recall that a set of rooks is \textbf{\textit{dominating}} if it guards every square of the polyomino.  Thus, a set of non-attacking rooks is maximal if and only if it is dominating.  The next lemma shows that the problem of finding the minimum size of maximal set of non-attacking rooks is the same as finding the minimum number of rooks to guard the polyomino, and thus is NP-hard as we have shown in Theorem~\ref{teo:nphard-rook}.

\begin{lemma}\label{lemma:nonattacking}
If there exists a set of $m$ rooks that dominates a polyomino $P$, then there exists a set of at most $m$ non-attacking rooks that dominates $P$.  
\end{lemma}

\begin{proof}
 Let $t_1, \ldots, t_m$ be the positions of the $m$ rooks that guard $P$.  We consider the rooks in order, moving them around or removing them so that at the $i$th step, the first $i$ rooks are non-attacking.  More precisely, we formally create a new ``off the board'' position that does not attack anything.  Then we find new positions $t'_1, \ldots, t'_m$ that may be off the board, with the property that for each $i$, the set of rooks at positions $t'_1, \ldots, t'_i, t_{i+1}, \ldots, t_m$ still guards $P$, and the set of rooks at positions $t'_1, \ldots, t'_i$ is non-attacking.
 
 To do this, suppose that we have chosen $t'_1, \ldots, t'_{i-1}$.  If $t_i$ is not in the same row or column as any of these earlier rooks, then set $t'_i = t_i$.  If $t_i$ is in the same row as one of the earlier rooks, and also in the same column as one of the earlier rooks, then set $t'_i$ to be off the board.  If $t_i$ is in the same row as one of the earlier rooks, but not in the same column as any of them, then we choose $t'_i$ to be a tile in the same column as $t_i$ that is not guarded by any of the earlier rooks, if such a position exists.  If every tile in the same column as $t_i$ is guarded by one of $t'_1, \ldots, t'_{i-1}$, then we set $t'_i$ to be off the board.  A similar procedure applies to the case where $t_i$ is in the same column as an earlier rook, but not in the same row as any of them.  In any of these cases we can cause $t'_1, \ldots, t'_i$ to be non-attacking while still guarding the same squares as before and maybe more.
 
 After we have finished this process, the rooks at positions $t'_1, \ldots, t'_m$, excluding those that are off the board, guard all of the polyomino $P$ and are non-attacking.
 \end{proof}

Note that this lemma statement is false for queens and for higher-dimensional rooks.  For instance, there are polycubes that can be guarded by two higher-dimensional rooks but these rooks must attack each other, and similarly for queens.

The next theorem shows that although by Theorem~\ref{teo:nphard-rook} it is NP-hard to determine whether $k$ non-attacking rooks can guard all the squares of a given polyomino, we can still use the bipartite graphs to say something about the values of $k$ for which this is possible.

\begin{theorem}\label{teo:interval}
Represent by $min(P)$ the minimum number of non-attacking rooks that dominate a polyomino $P$ and by $max(P)$ the maximum number of non-attacking rooks that can be placed on $P$. Then, $min(P) \geq \ceil{max(P)/2}$ and for any $min(P) \leq k \leq max(P)$ there exists a dominating set of $k$ non-attacking rooks on $P$.
\end{theorem}

\begin{proof}
 The statement $min(P) \geq \ceil{max(P)/2}$ follows from the well-known fact that maximal matchings differ in size by at most a factor of $2$.  Specifically, let $M_{max}$ be a matching of size $max(P)$ in the bipartite graph $G(P)$, and let $M_{min}$ be a maximal matching of size $min(P)$.  (The property of dominating the polyomino corresponds to being a maximal matching.)  Then every edge of $M_{max}$ shares at least one vertex with an edge in $M_{min}$, in order not to contradict the maximality of $M_{min}$; there are $2\cdot min(P)$ vertices of $M_{min}$, so there are at most $2\cdot min(P)$ vertices in $M_{max}$ and so $min(P) \geq \ceil{max(P)/2}$.
 
 To show that there is a dominating set of $k$ non-attacking rooks for every $k$ between $min(P)$ and $max(P)$, we start with a maximal matching $M_{min}$ of size $min(P)$ and apply an augmenting path algorithm.  This algorithm repeatedly increases the size of the matching by $1$ until it reaches size $max(P)$; we claim that at each step, the matching remains maximal.  Specifically, consider any maximal matching $M$.  Being maximal means that every edge of $G(P)$ is incident to an edge of $M$.  An augmenting path has odd length and alternates between edges not in $M$ and edges in $M$, starting and ending with edges not in $M$.  The matching $M'$ is obtained from $M$ by swapping whether each edge of the augmenting path is in $M$ or not in $M$.  We observe that $M'$ is still maximal, because every vertex of $M$ is still a vertex of $M'$.  Thus the augmenting-path algorithm preserves maximality, and the intermediate stages between $M_{min}$ and a maximum-size matching correspond to dominating sets of $k$ non-attacking rooks for all $k$ between $min(P)$ and $max(P)$.
\end{proof}



If we know that a polyomino admits a given number of non-attacking or dominating rooks, a natural question is to ask when that configuration is unique.  We show in Theorem~\ref{teo:odd-unique} that the maximum-size set of non-attacking rooks is unique in the special case that the polyomino has the smallest possible size that admits a given number of non-attacking rooks.  First we show in Theorem~\ref{teo:max-non-attacking} that if there are $m$ non-attacking rooks, then the polyomino has at least $2m-1$ tiles.  We use Lemmas~\ref{lem:connected} and~\ref{lemma:sizematching} to prove Theorem~\ref{teo:max-non-attacking}.

\begin{lemma}\label{lem:connected}
Given a polyomino $P$, the corresponding graph $G(P)$ is connected.
\end{lemma}

\begin{proof}
 Consider an equivalence relation on the tiles of $P$, generated by the relation that any two tiles that touch along an edge are equivalent.  Because the interior of $P$ is connected, all tiles of $P$ are in the same equivalence class.  Then we observe that any two tiles that touch along an edge are certainly in the same row or column.  Thus, if we generate an equivalence relation by the relation that any two tiles in the same row or column are equivalent, then it is also (even more) true that all tiles of $P$ are in the same equivalence class.  
 
 The corresponding construction in $G(P)$ is to require that any two edges that share a vertex are equivalent.  In this case, the connectedness of $P$ implies that all edges in $G(P)$ are in the same equivalence class.  Because any two edges that share a vertex are in the same connected component, this implies that all edges in $G(P)$ are in the same connected component.  Every vertex in $G(P)$ is in some edge, so $G(P)$ must be connected.
 \end{proof}

\begin{lemma}\label{lemma:sizematching}
In a connected graph $G$ with a matching $M$ of size $m$ there must be at least $m-1$ edges not in the matching.
\end{lemma}
\begin{proof}
If $M$ is a matching of size $m$, then the number of vertices in $G$ is at least $2m$. Then, the connectivity of $G$ implies that $G$ has at least $2m-1$ edges. Thus, $G$ has at least $m-1$ edges that are not in $M$.
\end{proof}

Using these two lemmas, we are ready to prove Theorem~\ref{teo:max-non-attacking}.

\begin{theorem}\label{teo:max-non-attacking}
The maximum number of non-attacking rooks that can be placed on a polyomino with $n$ tiles is $\ceil{\frac{n}{2}}$.
\end{theorem}

\begin{proof}
The polyominoes in the family constructed in Theorem~\ref{teo:minimum-rook} to require $\floor{\frac{n}{2}}$ rook guards also admit $\ceil{\frac{n}{2}}$ non-attacking rooks.  To show that no polyomino admits more non-attacking rooks, let $m$ be the maximum number of non-attacking rooks on a given polyomino $P$ with $n$ tiles.  Then the bipartite graph $G(P)$ has a matching of size $m$, so by Lemma~\ref{lemma:sizematching} there are at least $m-1$ edges of $G(P)$ not in the matching.  Thus $G(P)$ has at least $2m-1$ edges.  We know that $G(P)$ has $n$ edges, so $n \geq 2m-1$ and $m \leq \ceil{\frac{n}{2}}$.
\end{proof}

In the special case with the minimum number of tiles for the number of rooks, the configuration is unique.

\begin{theorem}\label{teo:odd-unique}
If a polyomino with an odd number of tiles $n=2m-1$ admits a set of $\ceil{\frac{n}{2}} = m$ non-attacking rooks, then this set dominates and is unique.
\end{theorem}

\begin{proof}
 Suppose for the sake of contradiction that the set of rooks does not dominate.  Then there is a square of the polyomino that is not attacked by a rook.  We can place an additional rook there to make a set of $m + 1$ non-attacking rooks, contradicting Theorem~\ref{teo:max-non-attacking}.
 
 For uniqueness, we show that a connected graph $G$ with $2m-1$ edges cannot have two different matchings of size $m$.  Suppose for the sake of contradiction that there are two matchings $M_1$ and $M_2$, both of size $m$.  Consider the subgraph $H$ of $G$ consisting of all edges that are either in $M_1$ but not $M_2$, or in $M_2$ but not $M_1$.  Because $G$ has only $2m$ vertices, each vertex of $G$ is incident to exactly one edge of each matching, so every vertex of $G$ is incident to either $0$ or $2$ edges of $H$.  Thus $H$ is a disjoint union of cycles and isolated points.  But we know that $G$ is a tree, because it is connected and has $2m$ vertices and $2m-1$ edges.  Thus $H$ cannot have any cycles in it, and so the two matchings are identical.
\end{proof}

For uniqueness, it is not enough that the number of non-attacking rooks be maximum given the number of tiles.  Some polyominoes with an even number of tiles $n=2m$ will have more than one set of $m$ non-attacking rooks. In Figure \ref{fig:leader} we gave an example of a polyomino with 10 tiles and with more than one set of $5$ non-attacking rooks. For instance, another four different configurations can be obtained if we select another rook to guard the center column. The same pattern gives a sequence of polyominoes with $n=2m$ tiles and more than one set of $m$ non-attacking rooks.

However, there are also some cases where a polyomino with $2m$ tiles admits only one set of $m$ non-attacking rooks.  In Figure~\ref{fig:even} we give a family of polyominoes with even number of tiles $n=2m$ and with a unique set of $m$ non-attacking rooks. 

\begin{theorem}\label{teo:even-unique}
There exists a sequence of polyominoes with $n=2m$ tiles, for $m \geq 3$, and only one set of $m$ non-attacking rooks.
\end{theorem}

\begin{proof}
The sequence of polyominoes with $n=2m$ tiles and only one set of $m$ non-attacking rooks is constructed as follows. The first three polyominoes, with $m=3, 4, 5$, and their respective sets of $m$ non-attacking rooks are depicted in Figure~\ref{fig:even}. For $m>3$, to recursively construct the polyomino of the sequence with $n=2(m+1)$ tiles from the polyomino with $n=2m$ tiles we add a vertical domino next to the right of the rightmost column of the polyomino, offset from the previous domino so that it creates a new row. The polyominoes constructed with this recursive algorithm have $n=2m$ tiles and $m$ columns. In order to have a set of $m$ non-attacking rooks, each column most have a rook placed on it. On each one of the polyominoes of the sequence, the third column, from left to right, has only one tile that we represent for the rest of the proof by $T$. Thus a rook must be placed on $T$. This forces the rest of the rooks that are on the columns to the right of $T$, to be placed on the tiles that are not in the same row as $T$. The tiles that are to the left of $T$ have the same configuration on each one of the polyominoes of the sequence. The leftmost column of each one of these polyominoes has only one tile that is not on the same row as $T$, thus, a rook must be placed there. This forces us to place a rook on the uppermost tile of the second column from left to right.
\end{proof}

\begin{figure}[h]
\begin{center}
\begin{tikzpicture} 
\foreach \x/\y in {  0/0, 0/1, 1/0, 1/1, 1/2, 2/1} { 
\path [draw=gray, fill=gray] (2.5+\x-0.45, .5+\y-0.45)
-- ++(0,.9)
-- ++(.9,0)
-- ++(0,-.9)
--cycle;
}

\foreach \x/\y in { 0/0, 1/2, 2/1} { 
 \node[anchor=west] at (2+\x+.05, \y+.5) {\rook};
}

\foreach \x/\y in {  0/0, 0/1, 1/0, 1/1, 1/2, 2/1, 3/1, 3/2} { 
\path [draw=gray, fill=gray] (7.5+\x-0.45, .5+\y-0.45)
-- ++(0,.9)
-- ++(.9,0)
-- ++(0,-.9)
--cycle;
}

\foreach \x/\y in { 0/0, 1/2, 2/1, 3/2} { 
 \node[anchor=west] at (\x+7.05, \y+.5) {\rook};
}

\foreach \x/\y in {  0/0, 0/1, 1/0, 1/1, 1/2, 2/1, 3/1, 3/2} { 
\path [draw=gray, fill=gray] (7.5+\x-0.45, .5+\y-0.45)
-- ++(0,.9)
-- ++(.9,0)
-- ++(0,-.9)
--cycle;
}

\foreach \x/\y in { 0/0, 1/2, 2/1, 3/2} { 
 \node[anchor=west] at (\x+7.05, \y+.5) {\rook};
}

\foreach \x/\y in {  0/0, 0/1, 1/0, 1/1, 1/2, 2/1, 3/1, 3/2, 4/0, 4/1} { 
\path [draw=gray, fill=gray] (13.5+\x-0.45, .5+\y-0.45)
-- ++(0,.9)
-- ++(.9,0)
-- ++(0,-.9)
--cycle;
}

\foreach \x/\y in { 0/0, 1/2, 2/1, 3/2, 4/0} { 
 \node[anchor=west] at (\x+13.05, \y+.5) {\rook};
}

\end{tikzpicture}
\end{center}
\caption{For $m \geq 3$, we depict in this figure the first three elements a sequence of polyominoes with $n=2m$ tiles that have a unique set of dominating non-attacking $m$ rooks. For $m=1$ and $m=2$ there does not exists a polyomino with $n$ tiles and a unique set of non-attacking $m$ rooks.}
\label{fig:even}
\end{figure}
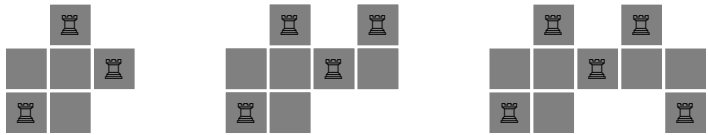

\section{Open Questions}

In~\cite{Worman07}, Worman and Keil prove that the $r$--visibility guard set problem is polynomial-time solvable when restricted to polyominoes without holes.

\begin{question}
Is the rook-visibility guard set problem for polyominoes still NP-hard if, in addition, we require that the polyomino be simply connected (i.e., not have holes)?
\end{question}

Another question is about the correspondence in Section~\ref{sec:nonattacking} between polyominoes and bipartite graphs.

\begin{question}
Which bipartite graphs can be written as $G(P)$ for some polyomino $P$?  For those that can, is there a good algorithm for reconstructing $P$?  Can we deduce anything about the number or size of holes in $P$ just from knowing $G(P)$?
\end{question}

For higher-dimensional rooks, the bipartite graph construction does not work.

\begin{question}
For $d$--dimensional rooks on $d$--polycubes, what are the analogues of the theorems in Section~\ref{sec:nonattacking}?  Can we prove similar results even without the bipartite graphs?
\end{question}

 Theorem~\ref{teo:interval} suggests that an entertaining recreational puzzle may be to give someone a polyomino $P$ and a number $k$, and to ask that person to produce a dominating set of $k$ non-attacking rooks on $P$.
 
 \begin{question}
 For which choices of $P$ and $k$ is the task of finding a dominating set of $k$ non-attacking rooks on $P$ of appropriate difficulty for a human?  What is a good method for producing such $P$ and $k$ by computer?
 \end{question}

We can also ask about random polyominoes; several probability distributions on polyominoes are given in~\cite{phdthesis}.  We know that the minimum number of rooks needed to guard a given $n$--omino ranges between $1$ (for a single-row polyomino) and $\floor{\frac{n}{2}}$ (as in Theorem~\ref{teo:minimum-rook}), and that the maximum number of non-attacking rooks ranges between $1$ and $\ceil{\frac{n}{2}}$.

\begin{question}
What is the expected value of the minimum number of rooks needed to guard a typical $n$--omino?  Is it bounded below by a linear function of $n$?  How about the maximum number of non-attacking rooks?
\end{question}

\bibliography{rooks}
\bibliographystyle{amsalpha}

\end{document}